\documentclass[11pt]{article}

\usepackage{amsmath}
\usepackage{amssymb}
\usepackage{latexsym}
\usepackage{graphicx}
\usepackage{amsthm}
\usepackage{color}
\usepackage{bbm}
\usepackage{float}

\usepackage{amsfonts}
\usepackage{bm}
\usepackage{ifthen}
\usepackage{cite}
\usepackage{caption}
\usepackage{subcaption}
\usepackage{mathtools}
\usepackage{footmisc}

\newcommand{\ignore}[1]{}

\newcommand{\jordan}[1]{$v^*_{#1}$}
\newcommand{\jordantime}[2]{$v^*_{\mathcal{#1}_{#2}}$}
\newcommand{\neighbor}[2]{$n^{#2}_{#1}$}
\newcommand{\neighborJordan}[2]{$n^{#2}_{v^*_{#1}}$}
\newcommand{\subtree}[2]{${#1}_{#2}$}
\newcommand{\subtreeJordan}[2]{$(#1,v^*_{#1})_{n^{#2}_{v_{#1}^*}} $}
\newcommand{\subtreeJordanTime}[3]{$(\mathcal{G}_{#1},v^*_{#2})_{n^{#3}_{v_{#2}^*}}$}

\newtheorem{theorem}{Theorem}[section]
\newtheorem{lemma}{Lemma}[section]
\newtheorem{definition}{Definition}[section]
\newtheorem{corollary}{Corollary}[section]

\title{Persistence of the Jordan center in Random Growing Trees}

\author{Sarath Pattathil, Nikhil Karamchandani and Dhruti Shah
\thanks{S. Pattathil is with the department of Electrical Engineering and Computer Science, Massachusetts Institute of Technology. E-mail: sarathp@mit.edu \newline
N. Karamchandani and D. Shah are with the department of Electrical Engineering, Indian Institute of Technology, Bombay. E-mail: nikhil.karam@gmail.com, dhruti96shah@gmail.com }
}

\date{}

\begin{document}

\maketitle

\begin{abstract}
The Jordan center of a graph is defined as a vertex whose maximum distance to other nodes in the graph is minimal, and it finds applications in facility location and source detection problems. We study properties of the Jordan center in the case of random growing trees. In particular, we consider a regular tree graph on which an infection starts from a root node and then spreads along the edges of the graph according to various random spread models. For the Independent Cascade (IC) model and the discrete Susceptible Infected (SI) model, both of which are discrete time models, we show that as the infected subgraph grows with time, the Jordan center persists on a single vertex after a finite number of timesteps. Finally, we also study the continuous time version of the SI model and bound the maximum distance between the Jordan center and the root node at any time. 
\end{abstract}

\section{Introduction}
\label{sec:Intro}
There are several notions of node centrality in graphs that have been proposed in the literature, such as distance centrality, betweenness centrality, degree centrality, and eigenvalue centrality (see for example \cite{Borgatti, Jackson}). These centrality measures find application in a wide variety of contexts, such as identifying influential/critical entities in social/communication networks \cite{Tan, Hwang}, source/root detection in diffusion/growing networks \cite{SZ, SIR, Fanti1, BDL}, and facility location problems \cite{JC, Handler, Slater}. 

Recently, there has been some work on analyzing the movement of graph centers in randomly growing graphs. In particular, the question of interest is whether a given notion of a graph center \textit{persists} in a random growth model, i.e., does a particular vertex emerge and remain as the center after a finite number of steps, even as the graph continues to grow or does the center move around indefinitely. \cite{Galashin} considered the preferential attachment model of growth \cite{BA}, where at each step a new node connects to existing nodes with probability proportional to the degrees of the existing nodes in the graph, and showed that with probability $1$ the degree center (node with the maximum degree) of the graph persists in this random growth model. In other words, after a finite number of timesteps, there is a single node which remains as the most-connected node in the network. In more recent work, \cite{Jog1} considered the centroid or `balancedness center' \cite{Mitchell} of a tree graph where the score of a given vertex is the maximum size of the subtrees rooted at its neighbors, and the center corresponds to the vertex with the minimum score. \cite{Jog1} proved the persistence of the balancedness center in both the uniform as well as preferential attachment tree growth models, as well as the random growing tree arising from an infection spread according to the popular Susceptible-Infected (SI) model on an underlying regular tree. Some follow-up work from the authors \cite{Jog2} studied the same question in sublinear preferential attachment trees. 

In this paper, we focus on the Jordan center of random growing trees. The Jordan center of a graph is defined as a vertex whose maximum distance to other nodes in the graph is minimal, and it finds applications in facility location \cite{JC, Handler, Slater} and source detection problems \cite{SIR, Ying2}. We consider an underlying regular tree on which an infection starts from a root node and spreads along the edges of the graph according to various random spreading models. For the Independent Cascade (IC) model and the discrete Susceptible Infected (SI) model, both of which are discrete time models, we show that as the infected subgraph grows with time, the Jordan center persists on a single vertex after a finite number of timesteps. 
Finally, we also study the continuous time version of the SI model and while we are unable to prove persistence in this case, we bound the maximum distance between the Jordan center and the root node at any time. 
To the best of our knowledge, there has been no prior study on the persistence properties of the Jordan center. In terms of previous results which are relevant to the contents of this paper, \cite{SIR} showed that when the infection spreads according to the discrete SI model on an underlying regular tree, the distance between the Jordan center of the infected subtree and the root node is bounded by a finite constant. We extend this result to demonstrate that under the discrete SI model, the Jordan center in fact reaches persistence in a finite number of timesteps, and thereafter does not move. 

Apart from being a very natural property to examine, persistence of centrality measures also has important implications for root-finding algorithms in growing graphs, which have applications such as identifying the center of an epidemic or the source of a rumor. It has been shown \cite{KhimLoh, BDL} that selecting the top nodes according to appropriate centrality measures can yield a confidence set for the root node such that it belongs to this set with high probability. Persistence of centrality measures in this context has been recently used \cite{Jog1} to show that the confidence set thus generated stabilizes after some finite time, implying that the construction is in some sense robust, which is a desirable property. This also suggests the possibility of savings in terms of computational costs, since while the size of the underlying graph and the corresponding complexity of running the root-finding algorithm increase with time, the output of the algorithm does not change beyond a certain threshold. We realize that while our results hold for trees, real-world networks are not tree-like and thus any application of these results to such settings would need an extension to more general network topologies. However, this is technically a really challenging problem and we consider the study of regular trees to be an important first step in this direction. Finally, while this paper mainly focuses on regular trees, some of our results do generalize to irregular trees and we discuss these briefly in Section~\ref{sec:Discussion}. 


The rest of the paper is organized as follows. We describe the graph and infection models we work with in Section~\ref{sec:Model}. Section~\ref{sec:prelim_results} describes some preliminary results for the Jordan center under the infection models we study. In Sections~\ref{sec:Jordan_IC} and \ref{sec:Jordan_DSI}, we show that the Jordan center is persistent in the IC model and the discrete SI models respectively. 
Section~\ref{sec:Jordan_SI} bounds the maximum distance between the Jordan center and the root node in the continuous time SI model. 
Section~\ref{sec:Simulation} provides some simulation results and Section~\ref{sec:Discussion} discusses generalizations of the results that we have proved and also some open problems.

\section{Problem Setup}
\label{sec:Model}
Let  $\mathcal{G} = (V(\mathcal{G}), E(\mathcal{G}))$ denote a tree graph. Define the function $\psi_\mathcal{G}: V(\mathcal{G}) \rightarrow \mathbb{N}$ as
\begin{align}
\psi_\mathcal{G}(u) = \text{depth} (\mathcal{G},u)
\end{align}
where $\text{depth}(\cdot, \cdot)$ is a function which takes as input a tree graph $\mathcal{G}$ and a vertex $u$, and returns the depth of $\mathcal{G}$ when rooted at $u$ i.e. the maximum distance of any  vertex from the root $u$. 

\begin{definition}
\label{def:Jordan}
A node \jordan{\mathcal{G}} is the Jordan center of a tree $\mathcal{G}$ if:
\begin{align}
\text{\jordan{\mathcal{G}}} \in \underset{v \in V(\mathcal{G})}{\operatorname{argmin}} \: \psi_\mathcal{G}(v)
\end{align}
and let
\begin{align}
\psi_\mathcal{G} = \psi_\mathcal{G}(\text{\jordan{\mathcal{G}}})
\end{align}
\end{definition}
$\psi_\mathcal{G}$ is called the \textit{centrality} of the Jordan center of tree $\mathcal{G}$. For any two nodes $u$ and $v$, if $\psi_\mathcal{G}(u) \leq \psi_\mathcal{G}(v)$, we say that $u$ is \textit{at least as central as} $v$. It can be easily verified that in a tree, there can be at most two Jordan centers and they have to be neighbors. 

\begin{figure}[th]
\begin{center}
\includegraphics[angle=0,scale=0.30]{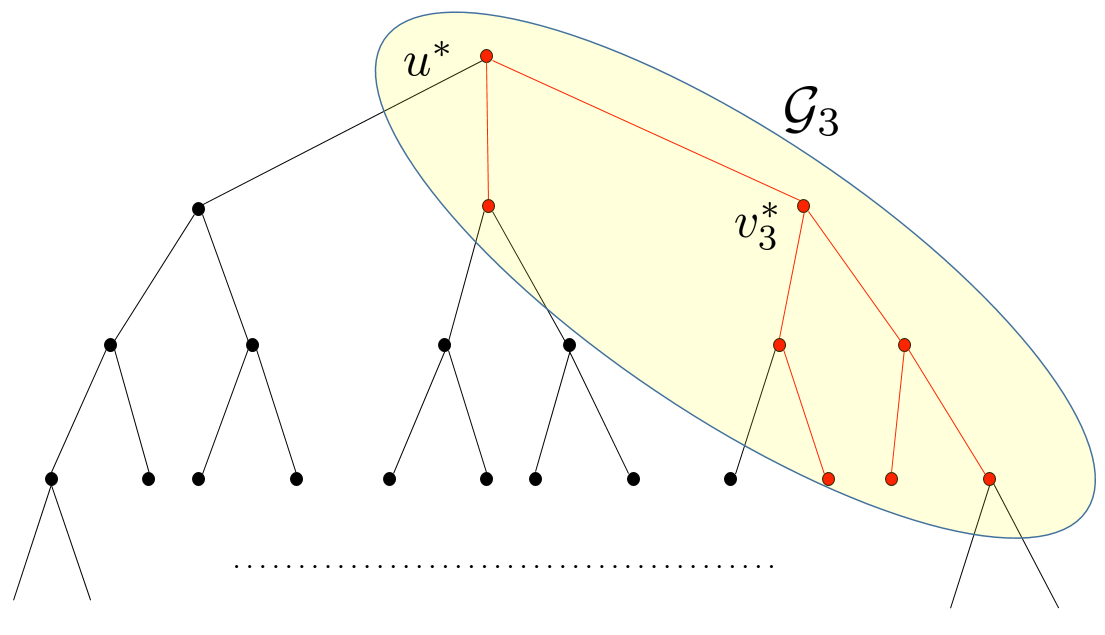}
\caption{$3$ regular tree. \jordan{3} denotes the Jordan center of the infected graph $\mathcal{G}_3$.}
\label{fig:(d+1)-regular_tree}
\end{center}
\end{figure}
The basic underlying tree model that we consider for our results is a $d+1$-regular tree for some $d\ge2$, where the root $u^{\star}$ has $d + 1$ children and every other vertex has exactly $d$ children. See Figure~\ref{fig:(d+1)-regular_tree} for an illustration. Starting from the root node $u^{\star}$, an infection spreads according to a random growth model (these will be discussed later in the section). At any time $t$, let the tree subgraph formed by infected nodes be denoted by $\mathcal{G}_t = (V_t,E_t)$. We study two properties of the Jordan center in random growing trees:
\begin{enumerate}
\item \textbf{Distance from the root:} We study how the distance of the root $u^{\star}$ to the Jordan center \jordan{\mathcal{G}_t} of the tree $\mathcal{G}_t$ changes with time $t$. In particular, can the distance  become very large or does the Jordan center remain close to the root at all times?
\item \textbf{Persistence of the Jordan center:} Here we study if the Jordan center is persistent. A Jordan center is said to be persistent if $\exists \ t_0 < \infty$ such that $\forall \ t \geq t_0$,  \jordantime{G}{t} = \jordantime{G}{t_0} i.e. the Jordan center does not change after time $t_0$. 
\end{enumerate}

\noindent We study three different random growth models in this paper:

\begin{definition}
\label{def:IC_Model}
\textbf{(Independent Cascade (IC) Model): }This is a discrete time model. The root node $u^{\star}$ is infected at time $0$. At each time step, an infected node infects each of its as yet uninfected neighbors independently with probability $p$. Each node can infect its neighbors for exactly one time step, after which it cannot infect and becomes sterile thereafter. We assume\footnote{If $pd \le 1$, the infected subtree will be finite under the IC model with probability $1$ and the problem is uninteresting.} $pd > 1$. 
\end{definition}

\begin{definition}
\label{def:DSI_Model}
\textbf{(Discrete Susceptible Infected (SI) Model): }This is a discrete time model. The root node $u^{\star}$ is infected at time $0$. At each time step, an infected node infects each of its as yet uninfected neighbors independently with probability $p$. An infected node can infect its neighbors at every following time step, unlike the case of the IC model. We assume $pd > 1$. 
\end{definition}

\begin{definition}
\label{def:SI_Model}
\textbf{(Susceptible Infected (SI) Model): }This is a continuous time model. The root node $u^{\star}$ is infected at time $0$. Each infected node can infect its as yet uninfected neighbors independently, and the time taken for the infection to spread along each edge is an independent and identically distributed random variable which is exponentially distributed with parameter $\lambda$ (which we take as $1$, without loss of generality). 
\end{definition}

\textit{Main results:} In a nutshell, the main results of this paper demonstrate the persistence of the Jordan center of an infected subtree growing according to the IC and the discrete SI models on an underlying $d+1$-regular tree. As corollaries, we show that for the IC and the discrete SI models, the distance between the root of the underlying tree and the Jordan center of the infected subtree is finite. 
Finally, we also  consider the the continuous time SI model. While we are unable to prove persistence of the Jordan center for this case, we show that if the depth of the infected subtree is $n$, then the distance between the root of the underlying tree and the Jordan center of the infected subtree is at most $O(\log n)$. 
While this paper mainly focuses on regular trees, some of our results do generalize to irregular trees and we discuss these briefly in Section~\ref{sec:Simulation}. 

We define some more notation which will be used throughout the paper. A rooted tree is denoted by $(\mathcal{G}, u)$, where $u \in V(\mathcal{G})$ is the root node. For a rooted tree $(\mathcal{G}, u)$, let \subtree{(\mathcal{G},u)}{v} denote the subtree rooted at $v$ which consists of all vertices $w \in V(\mathcal{G})$ such that the path from $u$ to $w$ passes through $v$. Alternatively, if we interpret the rooted tree  $(\mathcal{G}, u)$ as a tree of descendants of the ancestor $u$,  then \subtree{(\mathcal{G},u)}{v} denotes the subtree rooted at $v$ and  consisting of $v$ and the descendants of $v$. For any positive integer $k$, let $[1:k]$ denote the set $\{1,2,\ldots,k\}$.

\section{Preliminaries}
\label{sec:prelim_results}
In this section, we will present a few general properties regarding the movement of the Jordan center in growing trees, which hold true for all the spread models we consider in this paper. For a node $v$ , define \neighbor{v}{i} to be the $i^{th}$ neighbor of $v$ in the rooted subtree $(\mathcal{G},v)$, for some ordering of the neighboring nodes. Without loss of generality, assume that the ordering of the neighbors is such that \neighbor{v}{1} is the neighbor which has the deepest subtree, $(\mathcal{G},v)_{\text{\neighbor{v}{1}}}$ and \neighbor{v}{2} is the neighbor which has the second deepest subtree, $(\mathcal{G},v)_{\text{\neighbor{v}{2}}}$. In case both the depths are the same, the numbering can be arbitrary. See Figure~\ref{fig:lemma_pre_2} for an illustration. 

\begin{lemma}
\label{lemma:3_two_subtrees}
Let the Jordan center of the tree $\mathcal{G}$ be the node \jordan{\mathcal{G}}. Then the depth of the deepest subtree \subtreeJordan{\mathcal{G}}{1} is $\psi_\mathcal{G}-1$ and the depth of the second deepest subtree \subtreeJordan{\mathcal{G}}{2} is either $\psi_\mathcal{G} - 1$ or $\psi_\mathcal{G}-2$.
\end{lemma}
\begin{proof}

The fact that the depth of \subtreeJordan{\mathcal{G}}{1} is $\psi_\mathcal{G}-1$ follows from the definition of the Jordan center in Definition~ \ref{def:Jordan}. The second part can be proved by contradiction. Suppose the second deepest subtree \subtreeJordan{\mathcal{G}}{2} has depth $\psi_\mathcal{G} - k$ for some $k > 2$. Then, on moving the Jordan center by $\lfloor k / 2 \rfloor$ in the direction of the deepest subtree, we reach a node whose centrality is $\psi_\mathcal{G} - \lfloor k / 2 \rfloor < \psi_\mathcal{G}$, which contradicts the fact that \jordan{\mathcal{G}} is the Jordan center.

\end{proof}
Consider a sequence of growing random trees $\{ \mathcal{G}_t, \ t \geq 0 \}$. For a discrete time model, we say that the Jordan center changes from time $t_0$ to time $t_0 + 1$, if $ \exists \: \text{\jordan{\mathcal{G}_{t_0+1}}} \in V(\mathcal{G}_{t_0+1}) \text{ s.t. } \psi_{\mathcal{G}_{t_0+1}}(\text{\jordan{\mathcal{G}_{t_0}}}) > \psi_{\mathcal{G}_{t_0+1}}(\text{\jordan{\mathcal{G}_{t_0+1}}})$ i.e. we assume that the center changes only when the centrality of the new center is \textit{strictly} lesser than the centrality of the current center. For ease of notation, define \jordan{t} as the Jordan center of the infection graph at time $t$, $\mathcal{G}_t$. 
\begin{figure}[th]
\begin{center}
\includegraphics[angle=0,scale=0.4]{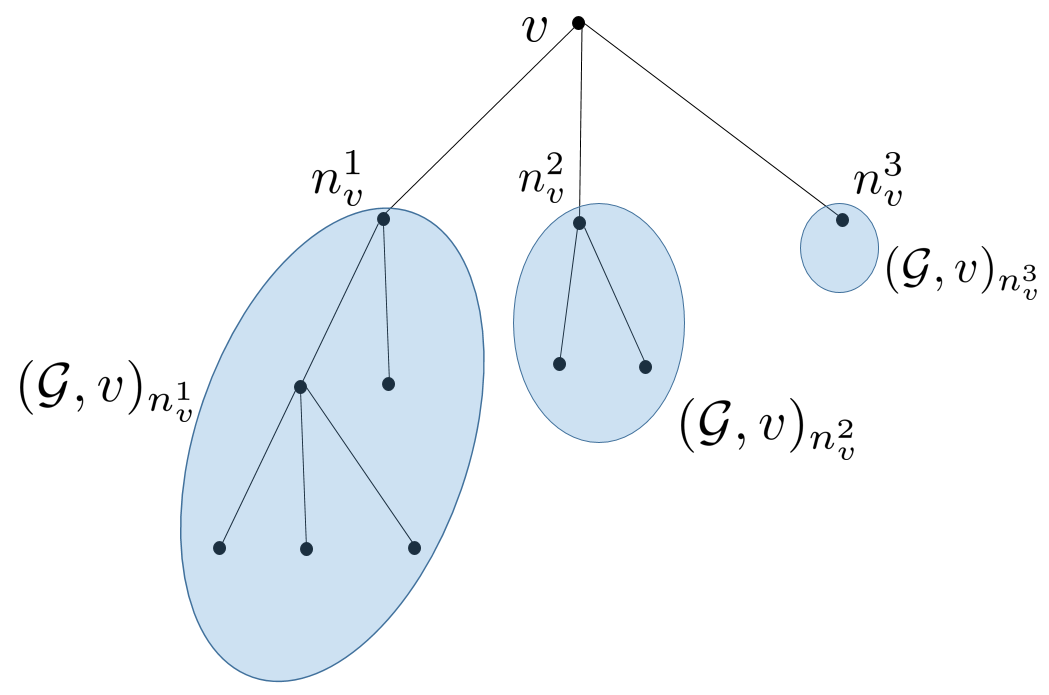}
\caption{Subtree rooted at node $v$. The neighbors and the subtrees corresponding to this root are shown in this figure.}
\label{fig:lemma_pre_2}
\end{center}
\end{figure}
\begin{lemma}
\label{lemma:dir_of_movement}
The Jordan center moves from time $t$ to time $t+1$ if and only if the second deepest subtree has a depth $\psi_{\mathcal{G}_t} - 2$ at time $t$, and at time step $t+1$, the depth of the second deepest subtree does not grow, while the deepest subtree grows by $1$. In other words, the Jordan center changes from time $t$ to $t+1$ if and only if  the following equations are satisfied:
\begin{align}
\text{depth}(\text{\subtreeJordanTime{t}{t}{2}}) =\text{depth}( \text{\subtreeJordanTime{t+1}{t}{2}}) = \psi_{\mathcal{G}_t} - 2 ,
\end{align}
\begin{align}
\text{depth}(\text{\subtreeJordanTime{t}{t}{1}}) &= \psi_{\mathcal{G}_t} - 1 , \\
\text{depth}(\text{\subtreeJordanTime{t+1}{t}{1}}) &= \psi_{\mathcal{G}_t} .
\end{align}
Furthermore, in case the center moves, the new Jordan center at time $t + 1$ will be the neighbor \neighborJordan{t}{1}. 
\end{lemma}
\begin{proof}
For sake of simplicity, we assume here that the deepest and second deepest subtrees at time $t$ are unique.  By this we mean that $\forall \: i \geq 3$
\begin{align*}
\text{depth}(\text{\subtreeJordanTime{t}{t}{i}}) < \text{depth}(\text{\subtreeJordanTime{t}{t}{2}}) .
\end{align*}
The proof can be easily adapted to the more general scenario.  From Lemma \ref{lemma:3_two_subtrees}, we know that the second deepest subtree \subtreeJordanTime{t}{t}{2} can only have two possible depths $\psi_{\mathcal{G}_t} - 1$ or $\psi_{\mathcal{G}_t} - 2$. Now, consider all  scenarios other than the one mentioned in the lemma: 
\begin{enumerate}
\item depth$\left(\text{\subtreeJordanTime{t}{t}{2}} \right) = \psi_{\mathcal{G}_t} - 1$: the second deepest subtree has depth $\psi_{\mathcal{G}_t} - 1$ at time $t$, 
\item depth$\left(\text{\subtreeJordanTime{t}{t}{2}} \right) = \psi_{\mathcal{G}_t} - 2$ and $\text{depth}\left(\text{\subtreeJordanTime{t+1}{t}{2}} \right) = \psi_{\mathcal{G}_t} - 1$: the second deepest subtree has depth $\psi_{\mathcal{G}_t} - 2$ at time $t$, but the depth grows at time $t+1$, and 
\item depth$\left(\text{\subtreeJordanTime{t}{t}{2}} \right)$ = depth$\left(\text{\subtreeJordanTime{t+1}{t}{2}} \right) = \psi_{\mathcal{G}_t} - 2$, and depth$\left(\text{\subtreeJordanTime{t}{t}{1}} \right)$ = depth$\left(\text{\subtreeJordanTime{t+1}{t}{1}} \right)$: the second deepest subtree has depth $\psi_{\mathcal{G}_t} - 2$ at time $t$, but neither the deepest nor the second deepest trees grow in depth at time $t+1$. 
\end{enumerate}
It can be verified that in all the above cases, the Jordan center will not move from time $t$ to $t+1$. For example, in case (1) the deepest and second deepest subtrees have the same depth at time $t$. Note that under all the growth models we consider in this paper, namely the IC, and the discrete SI models, the depth of either subtree can increase by at most $1$ at time $t+ 1$. Thus, no node can have a strictly lower centrality than the Jordan center at time $t$ and therefore, the center will not move at time $t+1$. 
\ignore{
Now, for all other cases, other than the one mentioned in the lemma, the Jordan center does not shift from time $t$ to $t+1$. The other cases are:

\begin{enumerate}
\item \subtreeJordanTime{t}{t}{2} has a depth $\psi_{\mathcal{G}_t} - 1$ at time $t$: In this case the Jordan center does not shift at time $t+1$, irrespective of which nodes get infected.
\item \subtreeJordanTime{t}{t}{2} has a depth $\psi_{\mathcal{G}_t} - 2$ at time $t$, but at time step $t+1$, the depth of this tree grows: In this case, the Jordan center does not shift irrespective of what happens to the neighboring nodes of any other subtree.
\item \subtreeJordanTime{t}{t}{2} has a depth $\psi_{\mathcal{G}_t} - 2$, but at time step $t+1$, the depth of \subtreeJordanTime{t+1}{t}{1} and \subtreeJordanTime{t+1}{t}{2} remains the same as that at time $t$: In this case, the Jordan center does not change irrespective of what happens to the other nodes in the tree.
\end{enumerate}
}

Next, we will argue that if the center moves at time $t+1$, it can only move in the direction of the deepest subtree to the neighbor \neighborJordan{t}{1}. See Figure~\ref{fig:lemma_pre_2}. Note that from time $t$ to $t+1$, the centrality of any node increases by at most $1$. Thus, it is evident that the Jordan center, if it shifts at time $t+1$, will move to one of the neighboring vertices $\left\{\text{ \neighborJordan{t}{i} }\right\}_{i \in [1:d+1]}$. 

Now say the Jordan center moves at time $t+1$, see Figure~\ref{fig:lemma_pre_2}. From the above argument, this implies that the second deepest subtree has a depth $\psi_{\mathcal{G}_t} - 2$ at time $t$, and at time step $t+1$, the depth of the second deepest subtree does not grow, while the deepest subtree grows by $1$.   For any $i \neq 1$, the centrality of node \neighborJordan{t}{i} at time $t+1$ will be more than the centrality of the original Jordan center \jordan{t} and therefore it cannot be the new Jordan center. It can be seen that if the Jordan center moves to node \neighborJordan{t}{1}, the centrality will reduce by $1$ as compared to the centrality of the original Jordan center \jordan{t}. This shows that the Jordan center will move here at time $t+1$.
\end{proof}

The above lemma gives rise to the following corollaries.
\begin{corollary}
\label{cor:center_change}
If at least two of the deepest subtrees rooted at the neighbors of the Jordan center grow in depth from time $t$ to $t+1$, the center will not change in this time step.
\end{corollary}

\begin{corollary}
\label{cor:centrality_change}
If the Jordan center shifts from time $t$ to $t+1$, we have
\begin{align*}
\psi_{\mathcal{G}_t} = \psi_{\mathcal{G}_{t+1}}.
\end{align*}
\end{corollary}

\section{Jordan center in the IC model}
\label{sec:Jordan_IC}
Recall the IC model from Definition~\ref{def:IC_Model}. In this section, we will consider the scenario where the infection starts from the root node $u^{\star}$ of a $d+1$-regular tree $\mathcal{G}$ and then spreads along the edges according to the IC model. For the sequence of infected trees $\{ \mathcal{G}_t, \ t \geq 0 \}$ growing according to the IC model, it is easy to see that all new vertices added at timestep $t$ will be at distance $t$ from the root node $u^{\star}$. This implies that as long as the infected tree is growing, the centrality of $u^{\star}$ at any time $t$ is $\psi_{\mathcal{G}_{t}}(u^{\star}) = t$. Also, we will say that a tree (or subtree) $\mathcal{T}$ is  `dead' at time $t_0$ if no new vertices are added to it at $t_0$, which implies the same for all $t > t_0$ as well because of the nature of the IC model. We begin with a few preliminary lemmas for the IC model, before presenting the main theorem of this section.
\begin{figure}[th]
\begin{center}
\includegraphics[angle=0,scale=0.4]{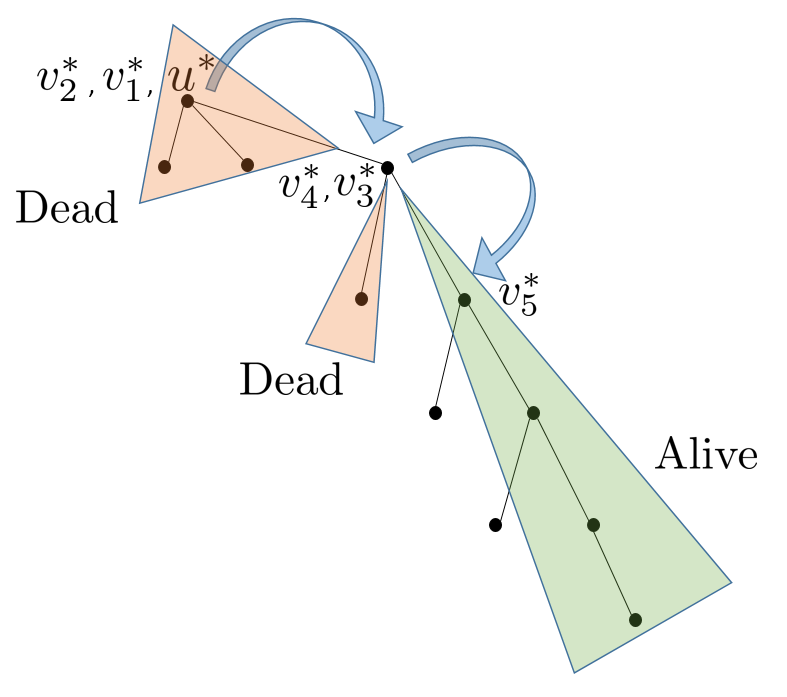}
\caption{The Jordan center changes from $u^*$ (which is also \jordan{1} and \jordan{2}) to \jordan{3} from time $2$ to $3$ and from \jordan{4} (which is the same as \jordan{3}) to \jordan{5} from timestep $4$ to $5$. Note that the center changes from \jordan{4} to \jordan{5} since all except one subtree of \jordan{4} are dead. Also note that \jordan{4} got infected at timestep $1$ and \jordan{5} got infected at timestep $2$.}
\label{fig:Jordan_IC}
\end{center}
\end{figure}
\begin{lemma}
\label{lemma:center_change_IC}
The Jordan center in the IC model changes from time $t$ to time $t+1$ if all except one of the subtrees rooted at the neighbors of the Jordan center \jordan{t} at time $t$ die, i.e., \subtreeJordanTime{t}{t}{i} dies for all $i \geq 2$. 
\end{lemma}
\begin{proof}
Consider the infected tree at time $0$, when only the root node $u^{\star}$ is infected. Trivially, $u^{\star}$ is the Jordan center at time $0$ denoted by $v^*_0$. Say the Jordan center shifts from the root at time $t_0 + 1$, for some $t_0 \ge 0$. From Lemma \ref{lemma:dir_of_movement}, the subtrees \subtreeJordanTime{t_0}{0}{i} must have depth $\leq (\psi_{\mathcal{G}_{t_0}}-2) = t_0 - 2$,  $\forall i \geq 2$ . From the definition of the IC model in Definition \ref{def:IC_Model}, this means that all the subtrees \subtreeJordanTime{t_0}{0}{i} for $i \geq 2$ are dead by time $t_0 - 1$. This proves the lemma for the first Jordan center i.e. the root node $u^{\star}$.

Assume the statement of the lemma is true for the first $n$ Jordan centers. Consider the $(n+1)^{th}$ Jordan center and say this center moves from time $t_n$ to $t_n + 1$. At time $t_n$, the subtrees for this Jordan center \jordan{t_n} are given by \subtreeJordanTime{t_n}{t_n}{i}. Again, using Lemma \ref{lemma:dir_of_movement}, we have that the center will change at time $t_n + 1$ if the deepest subtree grows by $1$, and the second deepest tree(s) does not grow, which for the IC model means that it is dead. For some $c \in [1:d+1]$, let \neighborJordan{t_n}{c} be the parent of \jordan{t_n} in the original tree $(\mathcal{G}, u^{\star})$ and thus, also the previous Jordan center since the center moves at most one hop at a time. Then from the induction hypothesis, we have that \subtreeJordanTime{t_n}{t_n}{c} is dead, see Figure~\ref{fig:Jordan_IC} for an illustration. Since the deepest subtree must have grown from time $t_n$ to $t_n + 1$, this implies that $c \neq 1$ and \subtreeJordanTime{t_n}{t_n}{c} is not the deepest subtree. Next, from the properties of the IC model, all the other growing subtrees amongst \subtreeJordanTime{t_n}{t_n}{i}, $i \neq c$ have the same depth. Since we know that the second deepest subtree did not grow from time $t_n$ to $t_n + 1$, this implies  that for the center to move, all subtrees \subtreeJordanTime{t}{t}{i}, $i \geq 2$ must be dead. This completes the proof.
\end{proof}

\begin{lemma}
\label{lemma:No_back}
If the Jordan center changes from time $t$ to time $t+1$, then the time of infection of \jordan{t} precedes that of \jordan{t+1}.
\end{lemma}
\begin{proof}
From Lemma \ref{lemma:center_change_IC}, we know that the Jordan center will move at time $t$ only if the deepest subtree grows, and all other subtrees are dead. From Lemma \ref{lemma:dir_of_movement}, in this case the new center will move in the direction of the deepest subtree to \neighborJordan{t}{1}. 
\ignore{
This implies.
\begin{align}
\text{depth(\subtreeJordanTime{t}{t}{1})} &= \psi_{\mathcal{G}_t} \\
\text{depth(\subtreeJordanTime{t+1}{t}{1})} &= \psi_{\mathcal{G}_t} + 1
\end{align}
}
For some $c \in [1:d+1]$, let \neighborJordan{t}{c} be the parent of \jordan{t} in the original tree $(\mathcal{G}, u^{\star})$. As argued in the proof of Lemma \ref{lemma:center_change_IC} above, the subtree \subtreeJordanTime{t}{t}{c} is dead at time $t$. Since the deepest subtree must have grown from time $t$ to $t + 1$, this implies that \subtreeJordanTime{t}{t}{c} is not the deepest subtree and thus the next Jordan center cannot be \neighborJordan{t}{c}. Finally, it is easy to see from the properties of the IC model that the time of infection of all the other neighbors of \jordan{t} is one more than \jordan{t}. This completes  the proof.
\end{proof}

\noindent The above lemma gives rise to the following corollary. 
\begin{corollary}
\label{cor:old_Jordan}
A node $v$ which was the Jordan center at time $t_0$, but shifted at time $t_0+1$, can never be the Jordan center for $t \geq t_0+1$.
\end{corollary}
We now present our main result for the Jordan center under the IC model. 
\begin{theorem}
\label{Thm:JCICP}
The Jordan center of an infected subtree growing according to the IC model on an underlying $d+1$-regular tree is persistent with probability 1.
\end{theorem}
\begin{proof}
Let the set $A$ be defined as follows:
\begin{align*}
A = \{ v : v = \text{\jordan{t}} \text{ for some time } t \}
\end{align*}
i.e. $A$ consists of all nodes that have been the Jordan center of the graph at some point of time. We will show that the size of the set $A$ is finite with probability $1$. This, along with Corollary \ref{cor:old_Jordan} will complete the proof.

Consider any node $v$ which gets infected, say at time $t_0 \ge 1$. Recall that $G_t$ denotes the infected subgraph at time $t$ and for any $t \ge t_0$, \subtree{(\mathcal{G}_t,u^{\star})}{v} denotes the subtree rooted at $v$ and consisting of $v$ and the infected descendants of $v$. It is easy to see that $\left\{\text{\subtree{(\mathcal{G}_t,u^{\star})}{v}}\right\}_{t \ge t_0}$ forms a Galton-Watson (GW) branching process \cite{Dawson}, which grows independent of the rest of the infection tree. A GW branching process is defined as follows. Let $\mathcal{Z}_0$ denote the number of nodes at time $0$. Then, the number of nodes at time $n$ is given by:
\begin{align*}
\mathcal{Z}_n = \sum_{i=1}^{\mathcal{Z}_{n-1}} \xi_i
\end{align*}
where $\xi_i's$ are i.i.d. random variables denoting the number of children spawned by each node in the previous generation. For our setup, $\mathcal{Z}_0 = 1$ and $\xi_i \sim bin(p,d)$. We say that a branching process is `dead' if  $ \exists$ $n_0$ such that $\mathcal{Z}_{n} = 0 \: \forall n \geq n_0$. For $pd > 1$, there is a positive probability, say $\bar{p}$, that the GW branching process will not die \cite{Dawson}. We  assume $pd > 1$ which implies that for any infected node $v$, there is a positive probability $\bar{p}$ that the rooted growing subtree $\left\{\text{\subtree{(\mathcal{G}_t,u^{\star})}{v}}\right\}_{t \ge t_0}$ will not die and grow forever. 

For any $i \ge 2$, let $\mathcal{E}_i$ be the  event $\mathcal{E}_i := \{ |A| \geq i \}$. Lemma \ref{lemma:center_change_IC} shows that the Jordan center \jordan{t} moves at the next step only when all except one of the subtrees $\left\{\text{\subtreeJordanTime{t}{t}{i}}\right\}_{i\in[1:d+1]}$ rooted at its neighbors are dead and the other subtree survives. Instead, consider the event $\mathcal{F}$ that for the Jordan center \jordan{t} all except one of the subtrees $\left\{\text{\subtreeJordanTime{t}{t}{i}}\right\}_{i\in[1:d+1]}$ rooted at its neighbors are dead and the remaining neighbor gets infected. For example, in Figure \ref{fig:Jordan_IC} all except one rooted subtrees of \jordan{4} die and the remaining neighbor gets infected, which in fact becomes the Jordan center  \jordan{5}. Note that  occurrence of the  event $\mathcal{F}$ is a necessary condition for the Jordan center \jordan{t} to move. Let the probability of this event be $q_0$ when the root node is the Jordan center, and $q$ when the center resides elsewhere. Since $\bar{p} > 0$ and $0 <p < 1$, we can see that $0<q_0, q<1$. Let 
$$
\mathcal{Q}_i := \{ \mbox{event $\mathcal{F}$ happens  at least $i-1$ times} \}.
$$ 
As can be seen from Figure \ref{fig:Jordan_IC} and Definition \ref{def:IC_Model}, the events $\mathcal{F}$ are independent across different Jordan centers, and thus we have
\begin{align*}
\mathbb{P}(\mathcal{Q}_i) = q_0 \times q^{i-2}  . 
\end{align*}
Also, we see that the event $\mathcal{E}_i$ happens only if $\mathcal{Q}_i$ happens, thereby showing that $\mathbb{P}(\mathcal{E}_i) \leq \mathbb{P}(\mathcal{Q}_i)$. Now,
$\sum_{i = 1}^{\infty} \mathbb{P}(\mathcal{E}_i) \leq \sum_{i = 1}^{\infty} \mathbb{P}(\mathcal{Q}_i) < \infty$ and thus from the Borel-Cantelli lemma, we have that with probability $1$, only finitely many of these events can happen. This shows that the size of the set $A$ is finite with probability $1$ and together with  Corollary \ref{cor:old_Jordan}, completes the proof of the theorem. 
\end{proof}
\begin{corollary}
\label{cor:JordanRootDistance}
The distance between the root of a $d+1$-regular tree and the Jordan center of an infected subtree growing according to the IC model on the underlying regular tree is finite.
\end{corollary}
\begin{proof}
This follows immediately from Theorem~\ref{Thm:JCICP} which shows that the number of distinct centers is finite, and the fact that the Jordan center in the IC model can move at most one hop at a time, from Lemma~\ref{lemma:dir_of_movement}.
\end{proof}

\section{Jordan Center in the discrete SI Model} 
\label{sec:Jordan_DSI}
%
Recall the discrete SI model from Definition ~\ref{def:DSI_Model}. In this section, we will consider the scenario where the infection starts from the root node $u^{\star}$ of a $d+1$-regular tree $\mathcal{G}$ and then spreads along the edges according to the discrete SI model. We borrow some notation and proof ideas from \cite{SIR}, which considered the maximum distance between the Jordan center and the root node. For the infected subtree $\mathcal{G}_t$ at any time $t$, we say that a node is at level $l$ if its distance from the root is $l$. Let $\mathcal{Z}_l$ denote the set of infected nodes\footnote{Note that the set $\mathcal{Z}_l$ grows over time, we have suppressed the time dependence in the notation for convenience.} at level $l$. $\mathcal{Z}_l^\tau$ is the set of infected nodes at level $l$, whose parents are in $\mathcal{Z}_{l-1}^\tau$ and who were infected within $\tau$ time slots after their parents were infected. This implies that all nodes in $\mathcal{Z}_{l}^\tau$ are infected by time $t \leq l \tau$. The sets $\mathcal{Z}_{0}$ and $\mathcal{Z}_{0}^\tau$ for any $\tau \geq 0$ are singletons and consist of the root node $u^{\star}$. Also, let $Z^\tau_l = |\mathcal{Z}_l^\tau|$. It is not difficult to see that the evolution of $\mathcal{Z}_l^\tau$ with level $l$ forms a Galton-Watson (GW) branching process \cite{GW}, which we denote by $B^\tau$ i.e. $B^\tau (l) = Z_l^\tau$. 

\begin{lemma}
\label{lemma:SIR_exist}
Given any $\epsilon > 0$, we can find sufficiently large $\tau$ and $l$, independent of time and the number of infected nodes, such that the probability that at least two $B^1$ branching processes starting from $\mathcal{Z}_l^\tau$ survive is at least $(1-\epsilon)$.
\end{lemma}
\begin{proof}
See proof of Theorem 5 in \cite{SIR}.
\end{proof}

\noindent The above lemma shows that starting from nodes at some finite distance $l$ from the root node $u^{\star}$, there are at least two subtrees in the infected subgraph, whose depths increase by $1$ at every time step. Next, we present the main result of this section.

\begin{theorem}
\label{Thm_Persistent_DSI}
Given any $\epsilon > 0$, the Jordan center of an infected subtree growing according to the discrete SI model on an underlying $d+1$-regular tree is persistent with probability at least $(1-\epsilon)$.
\end{theorem}
\begin{proof}
We prove this result by showing that $ \exists$ $t_0 < \infty$ such that the Jordan center of the infection subtree at time $t_0$, \jordan{t_0},  satisfies the following property: the two deepest subtrees  \subtreeJordanTime{t_0}{t_0}{1} and \subtreeJordanTime{t_0}{t_0}{2} rooted at neighbors of \jordan{t_0} increase in depth by one $\forall t \geq t_0$. From Corollary \ref{cor:center_change}, this implies that the Jordan center will not move from \jordan{t_0} and thus establishes its persistence.
 
Lemma \ref{lemma:SIR_exist} tells us that with probability at least $(1-\epsilon)$, there are at least two surviving $B^1$ processes starting from $\mathcal{Z}_l^\tau$, i.e., their depths increase by $1$ at each timestep. Let $d_t$ denote the distance from the Jordan center at time $t$, \jordan{t},  to the farthest node in one of the $B^1$ processes. As defined earlier, the centrality of the tree $\mathcal{G}_t$ is given by $\psi_{\mathcal{G}_t}$. Define the function $f(t)$, for $t \geq l\tau$, as follows:
\begin{align}
\label{Eqn:Function1}
f(t) = \psi_{\mathcal{G}_t} - d_t  .
\end{align}
Some of the properties of $f$ are:
\begin{enumerate}
\item \textit{$f(l\tau) \leq l\tau$}: Consider time $l \tau$. Starting from the root node $u^{\star}$, the depth of the infected subtree can increase by at most one in each time unit, and thus the centrality  of the root node $\psi_\mathcal{G}(u^{\star})$ is at most $l \tau$. From the definition of the Jordan center, this implies that the centrality of the Jordan center $\psi_{\mathcal{G}_t} \le l\tau$. We also have $d_{l\tau} \geq 0$ which gives the claimed property.
\item \textit{$f$ is non-increasing}: From time $t$ to $t+1$, if the Jordan center does not change, $d_{t+1} = d_t + 1$ since the $B^1$ process increases in depth by one at each timestep. Also, $\psi_{\mathcal{G}_{t+1}} \le \psi_{\mathcal{G}_t} + 1$ and thus $f$ is non-increasing in this case. 
On the other hand, if the center changes from time $t$ to $t+1$, we have $\psi_{\mathcal{G}_{t+1}} = \psi_{\mathcal{G}_t}$ from Corollary \ref{cor:centrality_change}.  Either the center moves in the direction of the $B^1$ process or away from it, which leads to $d_{t+1} = d_t$ and $d_{t+1} = d_t + 2$ respectively. In both cases, $f$ is non-increasing. 
\item \textit{$f(\cdot) \geq 0$}: Follows from the definition of the Jordan center, see Definition \ref{def:Jordan}.
\end{enumerate}

Recall that our goal is to show that after some finite time, the Jordan center will be such that the two deepest subtrees rooted at its neighbors increase in depth at each timestep. We begin by proving that $\exists \ \hat{t} < \infty$ such that $\forall \ T \ge \hat{t}$, the Jordan center, \jordan{T}, is such that for any $t\ge T$, the deepest rooted subtree \subtreeJordanTime{t}{T}{1} increases in depth by $1$ from $t$ to $t+ 1$. 

Consider the Jordan center \jordan{T} at some time $T \ge l\tau$. If the height of the deepest subtree, \subtreeJordanTime{t}{T}{1} increases by $1$ at every time step after $T$, we stop. If not, then $\exists \: t_1 < \infty$ such that the deepest subtree \subtreeJordanTime{t_1}{T}{1} does not increase in height at time $t_1 + 1$. Then from \eqref{Eqn:Function1}, we have 
\begin{align}
\label{eq:decrease}
f(t_1+1) - f(t_1)  &= \psi_{\mathcal{G}_{t_1+1}} - \psi_{\mathcal{G}_{t_1}} + d_{t_1}  - d_{t_1+1} \nonumber \\
&\overset{(a)}{=} d_{t_1}  - d_{t_1+1} \nonumber \\
&\overset{(b)}{=} -1,
\end{align}
where $(a)$ follows from Lemma~\ref{lemma:dir_of_movement}, which dictates that since the deepest subtree\footnote{The deepest rooted subtree for the Jordan center \jordan{t_1} also does not grow since \jordan{t_1} belongs to \subtreeJordanTime{t}{T}{1}, from Lemma~\ref{lemma:dir_of_movement}} does not grow, the Jordan center will not move from time $t_1$ to $t_1+1$, and furthermore, the centrality will remain the same; and $(b)$ follows from the fact that the Jordan center does not move from time $t_1$ to $t_1+1$, while the depth of the $B^1$ process increases by $1$.  Following $t_1 + 1$, we continue the process and again wait for the next timestep when the deepest subtree does not grow in depth. 

Let $\{t_1, t_2, \ldots\}$ denote the timesteps when the function $f(\cdot)$ decreases by $1$. Thus, we have
$$f(t_i) = l\tau - i .$$
Since $l \tau < \infty$, we have that the process will stop at some $\hat{t} < \infty$ when either the height of the deepest subtree \subtreeJordanTime{t}{\hat{t}}{1} increases by $1$ at every time step following $\hat{t}$ or $f(\hat{t}) = 0$. If it is the former, our claim is proved. Say it is the latter, then from the properties of the function $f$ detailed above, we have that $f(t) = 0$ $\forall t \ge \hat{t}$. From \eqref{Eqn:Function1}, this implies that $\psi_{\mathcal{G}_{t}}  = d_{t}$, which shows that the farthest node in one of the $B^1$ processes, is also the farthest node (or at least one of them) for the Jordan center \jordan{t} at any time $t \ge \hat{t}$. Since the $B^1$ process grows in depth by $1$ at each timestep, we have that $\forall \ T \ge \hat{t}$, the Jordan center, \jordan{T}, is such that for any $t\ge T$, the deepest rooted subtree \subtreeJordanTime{t}{T}{1} increases in depth by $1$ from $t$ to $t+ 1$. 

Now, what remains is to show that $\exists \ \tilde{t}$ such that $\hat{t} \le \tilde{t} < \infty$ and $\forall \ T' \ge \tilde{t}$, the Jordan center, \jordan{T'}, is such that for any $t \ge T' \geq \tilde{t}$, the second deepest rooted subtree \subtreeJordanTime{t}{T'}{2} also increases in depth by $1$ at every time step. Recall that Lemma \ref{lemma:SIR_exist} shows that with probability at least $(1-\epsilon)$, there are at least $2$ surviving $B^1$ processes starting  from $\mathcal{Z}_l^\tau$. Thus, we know that there is at least $1$ $B^1$ process which is not the deepest $B^1$ process considered in the first part of the proof above. Let this process be called $B^1_2$. 

We split the proof into 2 cases:

\textbf{Case 1:} When $B^1_2$ is a branch of the deepest subtree, see Figure \ref{fig:Case_1}. 
\begin{figure}[th]
\begin{center}
\includegraphics[angle=0,scale=0.4]{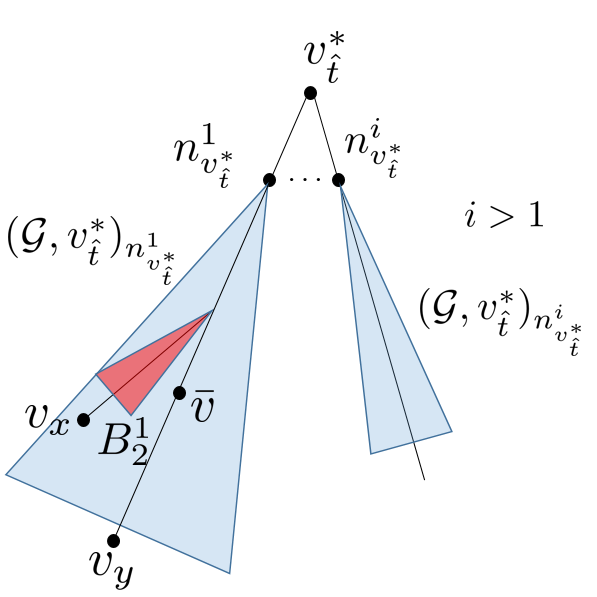}
\caption{Case 1}
\label{fig:Case_1}
\end{center}
\end{figure}
\noindent Consider time $\hat{t}$ as derived in the  proof of the first claim above. Consider the Jordan center \jordan{\hat{t}} and let the farthest node in $B^1_2$  be $v_x$ and the farthest node in the deepest subtree be $v_y$, see Figure \ref{fig:Case_1}. Consider a node $\bar{v}$ in the deepest subtree \subtreeJordanTime{\hat{t}}{\hat{t}}{1} such that:
\begin{align}
dist(\bar{v},v_x) = dist(\bar{v}, v_y).
\end{align}
From the definition of a Jordan center in Definition \ref{def:Jordan}, we have
\begin{align}
dist(\text{\jordan{\hat{t}}}, \bar{v}) \leq \psi_{\mathcal{G}_{\hat{t}}} < \infty .
\end{align}
Let $g(\cdot)$ be defined as follows:
\begin{align}
g(t) = dist(\text{\jordan{t}}, \bar{v}), \ \forall t \ge \hat{t}.
\end{align} 
Now, consider time $\hat{t}$ and the second deepest subtree \subtreeJordanTime{\hat{t}}{\hat{t}}{2} rooted at a neighbor of the Jordan center \jordan{\hat{t}}. If this subtree continues to grow in depth by $1$ at every following time step, then we are done. If not, $\exists \: \bar{t} < \infty$ such that  it does not grow at time $\bar{t}+1$. If this happens, then from Lemma \ref{lemma:dir_of_movement} we have that the Jordan center will change\footnote{Note that it is possible that the depth of the second deepest is the same as the deepest and in this case from lemma \ref{lemma:dir_of_movement}, the center will  change only when the second deepest does not grow for two distinct time steps. This technicality does not alter the flow of the proof.} from time $\bar{t}$ to $\bar{t} + 1$.  Furthermore, since the Jordan center moves in the direction of the deepest subtree, we have:
\begin{align}
g(\bar{t}+1) - g(\bar{t}) = -1.
\end{align}
Now, continue this process and stop when either the height of the subtree \subtreeJordanTime{t}{t}{2} grows by $1$ at every following time step, or when $g(\cdot)=0$. Let this time be $t_0'$. If $g(\cdot) = 0$, the Jordan center has reached $\bar{v}$, for which the two deepest subtrees are (a subtree of) the deepest subtree at time $\hat{t}$ and $B^1_2$, both of which grow by one at every following time step. Thus we have shown that after some finite time, the Jordan center will be such that the two deepest subtrees rooted at its neighbors increase in depth at each timestep. This completes the proof of the claim in Case 1. 

\begin{figure}[th]
\begin{center}
\includegraphics[angle=0,scale=0.4]{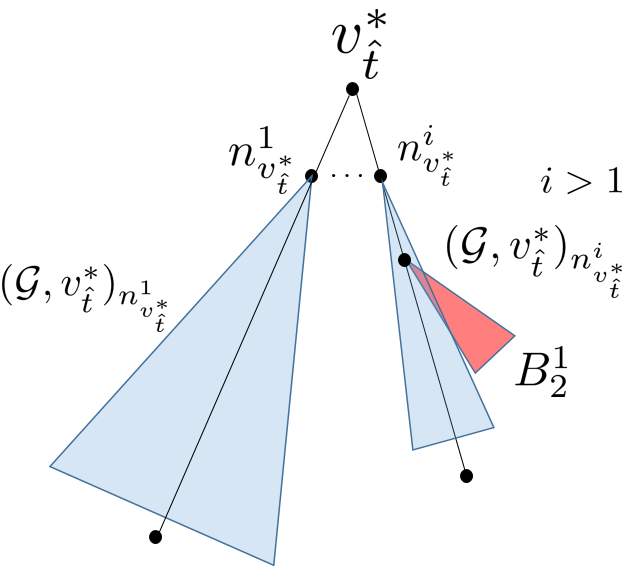}
\caption{Case 2}
\label{fig:Case_2}
\end{center}
\end{figure}
\textbf{Case 2:} When $B^1_2$ is not a branch of the deepest subtree, see Figure~\ref{fig:Case_2} for an illustration. The proof for this case is similar to the proof for the deepest subtree considered earlier, and so we skip it here. 

Finally, tying together all the pieces described above, we have that there exists some time $t_0<\infty$ such that $\forall t > t_0$, the two deepest subtrees \subtreeJordanTime{t}{t}{1} and \subtreeJordanTime{t}{t}{2} rooted at the neighbors of the Jordan center \jordan{t} increase in depth from $t$ to $t+1$. From Corollary \ref{cor:center_change}, this implies that the Jordan center will not move and thus establishes its persistence.
\end{proof}

\begin{corollary}
\label{cor:JordanRootDistance2}
Given any $\epsilon > 0$, the distance between the root of a $d+1$-regular tree and the Jordan center of an infected subtree growing according to the discrete SI model on the underlying regular tree is finite with probability at least $(1-\epsilon)$.
\end{corollary}
\begin{proof}
This follows immediately from Theorem~\ref{Thm_Persistent_DSI} which shows that the number of distinct centers is finite, and the fact that the Jordan center in the SI model can move at most one hop at a time, from Lemma~\ref{lemma:dir_of_movement}. This result has also been proved in \cite[Theorem 5]{SIR}.
\end{proof}
%



\section{Jordan center in the SI model}
\label{sec:Jordan_SI}
Recall the continuous time SI model from Definition ~\ref{def:SI_Model}. In this section, we will consider the scenario where the infection starts from the root node $u^{\star}$ of a $d+1$-regular tree $\mathcal{G}$ and then spreads along the edges according to the SI model. In this case, we show that when the height of the tree is $n$, the Jordan center is within a distance of $O(\log n)$ from the root $u^{\star}$ of the tree. 

\begin{lemma}
\label{lemma:macdiarmid}
For any $\epsilon > 0$ and $n$ large enough, when the depth of a deepest subtree $(\mathcal{G}_{t},u^{\star})_{n^{1}_{u^{\star}}}$ rooted at a neighbor of  $u^{\star}$ is $n$, the depth of every other subtree $(\mathcal{G}_{t},u^{\star})_{n^{i}_{u^{\star}}}$, $i >1$, is at least  $n- c_3 \log n $ for some constant $c_3 > 0$ with probability at least $(1 - \epsilon)$.
\end{lemma}
\begin{proof}

Consider the $d+1$ subtrees of the root $u^\star$. Each of them gets infected according to an independent SI model. For some $i \in [1:d+1]$, consider any neighbor of the root \neighbor{u^\star}{i} and let $B_n^i$ denote the first time a node at distance $n$ from the root node \neighbor{u^\star}{i} gets infected, i.e., $\text{depth} (\mathcal{G}_{B_n^i},u^{\star})_{n^{i}_{u^{\star}}} =  n$. Now, let the time of infection of \neighbor{u^\star}{k}, $k \in [1:d+1]$ be denoted by the random variable $T_k$. Since this is the SI model, we have $T_k \sim exp(1)$. Therefore, the total time for the infection to reach level $(n+1)$ through \neighbor{u^\star}{i} is $T_i + B_n^i$. We use results on branching random walks from \cite{mcdiarmid} to show that 
\begin{align}
\label{Eqn:BRW}
\mathbb{P}( B_n^i \leq \gamma n + c_1 \log n - x) \leq e^{-\delta x}
\end{align}
for some constants $\gamma, c_1, \delta>0$, details are provided below in Section~\ref{sec:concineqbrw}. For any $n\ge 1$, let $\mathcal{E}_n^i$ denote the event that $B_n^i \leq \gamma n + (c_1 - \frac{2}{\delta})\log n \overset{\Delta}{=} t_1$. Taking $x = \frac{2}{\delta}\log n$ in \eqref{Eqn:BRW}, we get
\begin{align}
\mathbb{P}\left( \mathcal{E}_n^i \right) \leq \frac{1}{n^2}. 
\end{align}
It is easy to see that $\sum_{n=1}^{\infty} \mathbb{P}\left( \mathcal{E}_n^i \right) < \infty$ and thus, applying the Borel-Cantelli lemma, we get that with probability $1$, only finitely many of these events can occur. Therefore, for $n$ large enough
\begin{align}
\label{Eqn:T1}
B_n^i > \gamma n + \tilde{c}_1\log n \overset{\Delta}{=} t_1 \qquad \text{w.p. }1
\end{align}
where $\tilde{c}_1 = c_1 - \frac{2}{\delta}$.

Next, consider some $j \neq i$ and the subtree rooted at the corresponding neighbor \neighbor{u^\star}{j}.  Let $B^j_{n - c_3 \log n}$ denote the first time a node at distance $n - c_3 \log n$ from the root node of the subtree, \neighbor{u^\star}{j},  gets infected, i.e., $\text{depth} (\mathcal{G}_{B_{n - c_3 \log n}}^j,u^{\star})_{n^{j}_{u^{\star}}} =  n - c_3 \log n$.  Again, using \cite{mcdiarmid} we also have
\begin{align}
\mathbb{P}(B^j_{n - c_3 \log n} \geq \gamma ({n - c_3 \log n}) +& c_2 \log ({n - c_3 \log n}) + x) \nonumber \\ 
&\leq e^{-\delta x}
\label{Eqn:BRW2}
\end{align}
for some $\gamma, c_2, \delta>0$ and any $c_3 \ge 0$. Once again, applying the Borel Cantelli lemma, we know that for $n$ large enough
\begin{align}
\label{Eqn:T2}
B^j_{n - c_3 \log n} \leq & \gamma (n - c_3 \log n) \nonumber \\
&+ \tilde{c}_2 \log ({n - c_3 \log n}) \overset{\Delta}{=} t_2 \qquad \text{w.p. }1
\end{align}
where $\tilde{c}_2 = c_2 + \frac{2}{\delta}$. 
\ignore{
Choose $n$ large enough and $c_3$ such that $t_2 < t_1$. Then, we have for large enough $n$ and any $i \neq j$:
\begin{align}
B^j_{n - c_3 \log n} < B^i_n \qquad \text{w.p. }1 .
\end{align}
}
The above arguments hold true for any neighbor \neighbor{u^\star}{j}, $j \neq i$. Next, note that by choosing $c_3$ large enough, one can make $t_2$ smaller than $t_1$. Now using \eqref{Eqn:T1} and \eqref{Eqn:T2}, we can choose $c_3$ large enough such that for $n$ large enough and any $j \ne i$, with probability at least $(1-\epsilon)$ we have:
\begin{align}
T_j - T_i \le t_1 - t_2 \ \Rightarrow \ B^j_{n - c_3 \log n} + T_j \le B^i_n + T_i.
\end{align}
This says that by the time any subtree $(\mathcal{G}_{t},u^{\star})_{n^{i}_{u^{\star}}}$ reaches a depth $n$, all other subtrees have depths greater than $n - c_3 \log n$ with probability $(1-\epsilon)$ and this proves the lemma. See Figure~\ref{fig:SI} for an illustration of the result.
\end{proof}
\noindent Using this lemma, we have the following theorem:
\begin{theorem}
\label{Thm:SIDistance}
Given any $\epsilon > 0$, when the depth of an SI infected tree is $n$, the Jordan center of the tree is within a distance of $C \log n$ from the root with probability at least $(1-\epsilon)$, for some constant $C$.
\end{theorem}
\begin{proof}
This follows from Lemma~\ref{lemma:macdiarmid}. When the depth of the tree is $n$, the original root, $u^\star$, has a centrality of $n$. Consider any node $v$ which is at a distance of more than $c_3 \log n$ from the root, where the constant $c_3$ has been defined in Lemma~\ref{lemma:macdiarmid}. From Lemma~\ref{lemma:macdiarmid}, for $n$ large enough the depth of all the subtrees rooted at the neighbors of the root node is at least $c_3 \log n$. Thus, the centrality of node $v$ is greater than $n$, which is the centrality of the root. This shows that the node with minimum centrality i.e. the Jordan center, has to lie within a distance of $c_3 \log n$ from the root. 
\begin{figure}[th]
\begin{center}
\includegraphics[angle=0,scale=0.5]{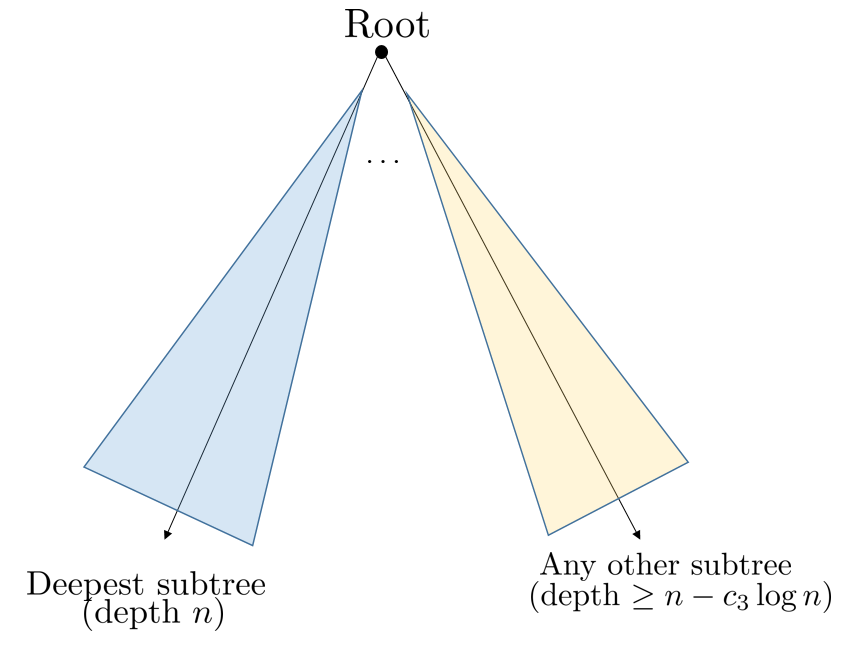}
\caption{Tree in the SI model}
\label{fig:SI}
\end{center}
\end{figure}
\end{proof}

\subsection{Proofs of \eqref{Eqn:BRW} and \eqref{Eqn:BRW2}}
\label{sec:concineqbrw}
In the proof of Lemma~\ref{lemma:macdiarmid}, what remains is to show how we obtain \eqref{Eqn:BRW} and \eqref{Eqn:BRW2}. We use Theorems $1$ and $2$ from \cite{mcdiarmid} in order to establish these inequalities. Here, we show that our problem setting satisfies all the necessary conditions for these theorems to be applied.
 
\noindent Let $Z_n(t)$ denotes the number of nodes $v$ in the $n^{th}$ generation (nodes which are at a distance of $n$ from the root node), which are born before time $t$, and let $Z_n = \sup_{t \rightarrow \infty} Z_n(t)$. Since our underlying tree model is a $d+1$-regular tree, we have $Z_n = d^n$. $B_n$ is defined as the first time of a birth in the $n^{th}$ generation. The time of infection of the nodes in the $n^{th}$ generation are denoted by $z_{n1}, z_{n2}, \cdots$. $F(t) = \mathbb{E}[Z_1(t)]$ and $\alpha = \inf \{t : F(t) > 0\}$. \\

\noindent Since we are working with the SI model, Definition \ref{def:SI_Model}, where the infection spread from one node to its neighbour is exponentially distributed with mean $1$, we can see that $\alpha = 0$. Now, define:
\begin{align}
\phi (\theta) = \mathbb{E}[ \sum_r e^{-\theta z_{1r}} ] .
\end{align}
Since $z_{1r} \sim exp(1)$ we have:
\begin{align}
\phi (\theta) = \frac{d}{1+\theta} .
\end{align}
Again, define:
\begin{align}
\mu (a) = \inf \{ e^{\theta a} \phi (\theta) : \theta \geq 0 \} .
\end{align}
On substituting $\phi(\cdot)$ we get:
\begin{align}
\mu(a) = ad e^{1-a} .
\end{align}
Time constant $\gamma$ is defined as:
\begin{align}
\gamma = \inf \{ a : \mu(a) \geq 1\} .
\end{align}
Since $\mu(0) = 0$ and $\mu(1) = d$, we have by the continuity of $\mu(\cdot)$ that 
\begin{align}
0 < \gamma < 1 .
\end{align}

\noindent We have that the conditions of Theorem 2 \cite{mcdiarmid} are satisfied which gives for some constants $\gamma > 0, c_1, \delta>0$
\begin{align}
\mathbb{P}( B_n \leq \gamma n + c_1 \log n - x) \leq e^{-\delta x} .
\end{align}
Since $Z_1 = d < \infty$, we apply part (b) of the theorem on $B_{n - c_3 \log n}$ to get:
\begin{align}
\mathbb{P}(B_{n - c_3 \log n} &\geq \gamma ({n - c_3 \log n}) \nonumber \\
&+ c_2 \log ({n - c_3 \log n}) + x) \leq e^{-\delta x} .
\end{align}

\section{Simulations}
\label{sec:Simulation}
In this section, we simulate infection spread on a regular tree using the IC model and the SI model and track the movement of the Jordan center of the infected subtree. In particular, we consider the maximum distance between the Jordan center and the root node and the number of times the center changes over the course of the simulation. 

\begin{figure}
        \begin{subfigure}[b]{0.49\textwidth}
                \includegraphics[width=\linewidth]{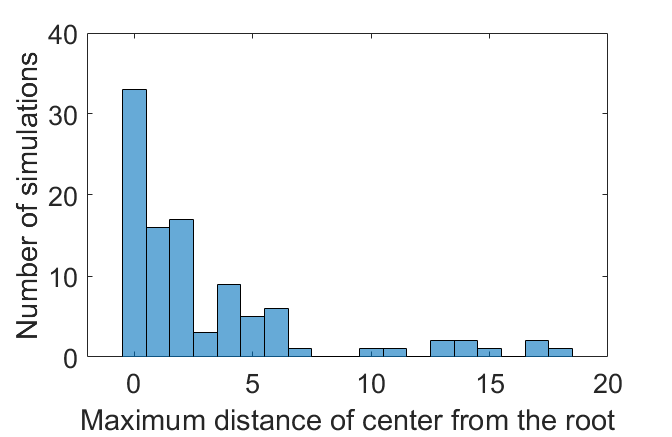}
                \caption{Distance from the root}
                \label{fig:Jordan_IC_dist}
        \end{subfigure}%
        \begin{subfigure}[b]{0.49\textwidth}
                \includegraphics[width=\linewidth]{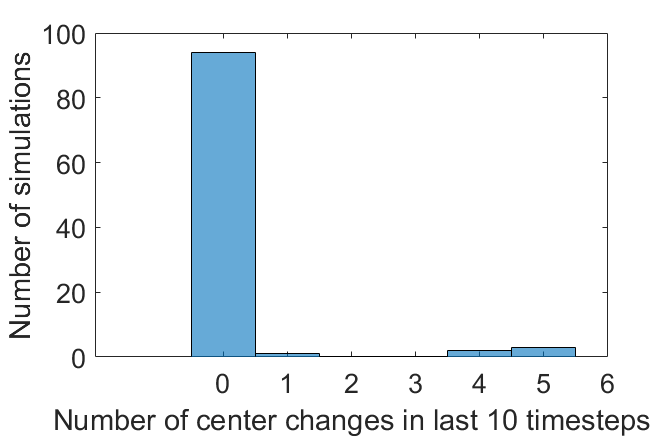}
                \caption{Number of  center changes}
                \label{fig:Jordan_IC_centre}
        \end{subfigure}%
        \caption{Jordan center in the IC model ($p = 0.4, d = 4$)}
\end{figure}

\begin{figure}
        \begin{subfigure}[b]{0.49\textwidth}
                \includegraphics[width=\linewidth]{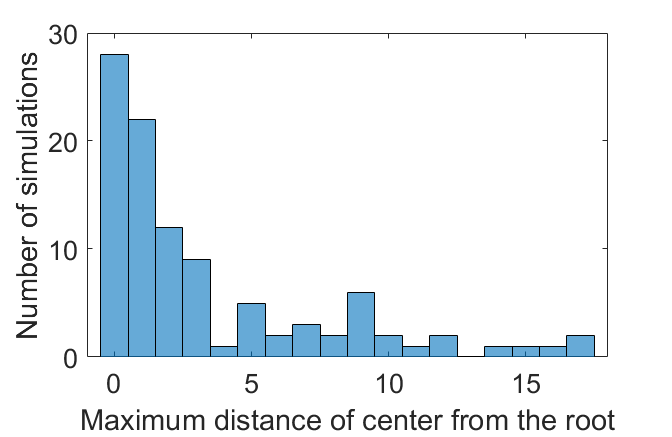}
                \caption{Distance from the root}
                \label{fig:Jordan_IC_dist_4}
        \end{subfigure}%
        \begin{subfigure}[b]{0.49\textwidth}
                \includegraphics[width=\linewidth]{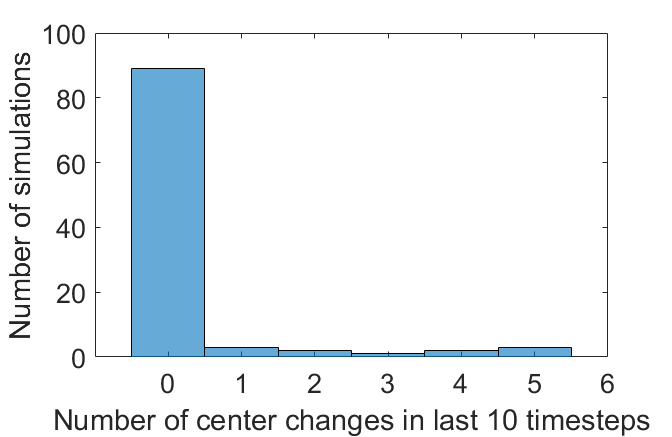}
                \caption{Number of  center changes}
                \label{fig:Jordan_IC_centre_4}
        \end{subfigure}%
        \caption{Jordan center in the IC model on an irregular tree ($p = 0.4, d \in \{3,4\}$)}
\end{figure}
\begin{figure}
        \begin{subfigure}[b]{0.49\textwidth}
                \includegraphics[width=\linewidth]{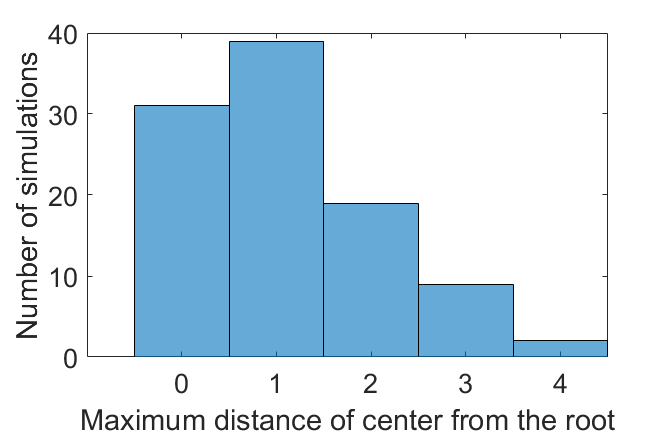}
                \caption{Distance from the root}
                \label{fig:Jordan_SI_dist}
        \end{subfigure}%
        \begin{subfigure}[b]{0.49\textwidth}
                \includegraphics[width=\linewidth]{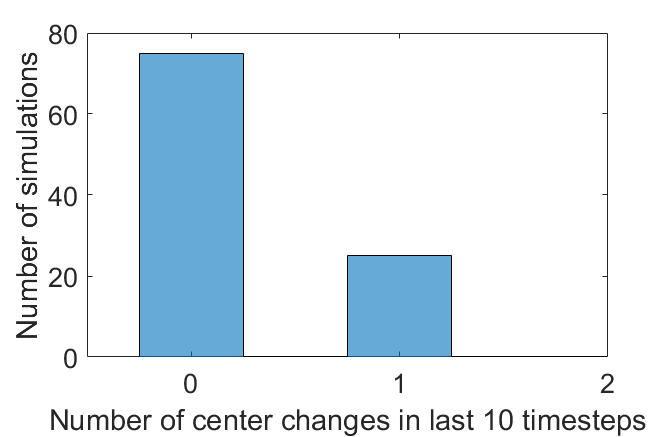}
                \caption{Number of distinct centers}
                \label{fig:Jordan_SI_centre}
        \end{subfigure}%
        \caption{Jordan center in the SI model ($d = 4$)}
\end{figure}

For the IC model, we created an underlying $4$-regular tree and used $p=0.4$ for the probability of an infected node spreading the infection to its neighbor. Starting with the root node, we spread the infection for $40$ timesteps\footnote{We were able to run the IC model simulation for only $40$ timesteps of the simulation because it became computationally infeasible beyond that. This is also why we were not able to run the simulations for the discrete SI model long enough and hence do not present the results here.} and evaluate the Jordan center of the infected subtree at each timestep. We repeated this experiment $100$ times and present our results in Figures~\ref{fig:Jordan_IC_dist}  and \ref{fig:Jordan_IC_centre}. We see that for a majority of the simulations, the Jordan center remains close to the root, as illustrated in Figure~\ref{fig:Jordan_IC_dist}. Also, from Figure~\ref{fig:Jordan_IC_centre} we see that in almost all the experiments, the Jordan center did not change  in the last $10$ out of the $40$ timesteps. This suggests that the Jordan center is close to attaining persistence, as suggested by our theoretical results in Section~\ref{sec:Jordan_IC}. 
 
For the continuous time SI model, we again considered a $4$-regular tree and in each run, ran the simulation till the infection spread to $100$ vertices. We tracked the movement of the Jordan center each  time a new node is added. We repeated the experiment $100$ times and present our results in Figures~\ref{fig:Jordan_SI_dist}  and \ref{fig:Jordan_SI_centre}. As for the IC model, the Jordan center remains quite close to the root over the course of the simulation, and the center only visits a few nodes during the last 30 out of the 100 timesteps. While our theoretical results for the continuous time SI model in Section~\ref{sec:Jordan_SI} did not prove persistence of the Jordan center or indeed even bounded distance from the root, empirical results strongly suggest these hold true.

\section{Discussion}
\label{sec:Discussion}
The theme of this paper broadly is how centers of an evolving graph change over time. We have looked at the specific case where an infection spreads on a regular tree following the IC and the  discrete SI models and we track the Jordan center of the infected subtree. There are several possible directions we can consider here:


\begin{figure}
        \begin{subfigure}[b]{0.49\textwidth}
                \includegraphics[width=\linewidth]{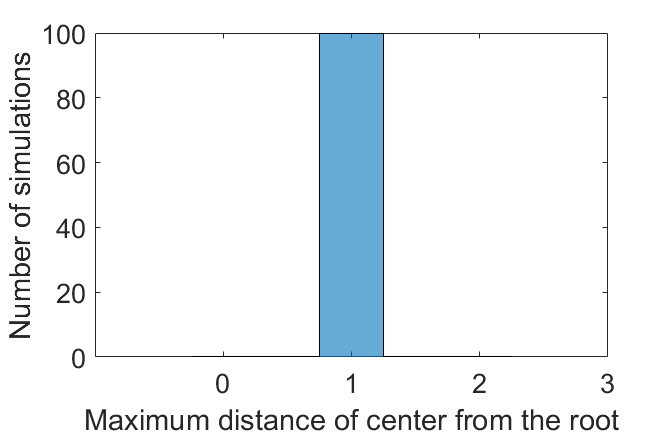}
                \caption{Distance from the root}
                \label{fig:Jordan_PA_dist}
        \end{subfigure}%
        \begin{subfigure}[b]{0.49\textwidth}
                \includegraphics[width=\linewidth]{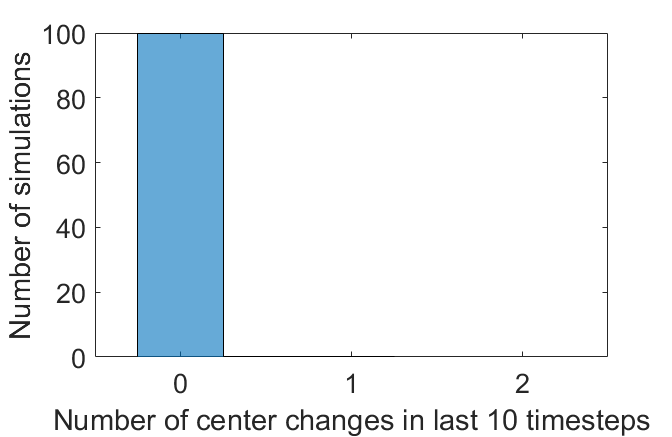}
                \caption{Number of center changes}
                \label{fig:Jordan_PA_centre}
        \end{subfigure}%
        \caption{Jordan center in the Preferential Attachment model}
\end{figure}

1) \textbf{Other graph topologies: } In this paper, we have considered infection spreads on regular trees.  Our results for the IC and the discrete SI models directly extend to irregular trees, with minimum degree of each node $d_{\min} > 3$, and $pd > 1$. This follows from the following observation: for any such irregular tree, we can consider the embedded $d_{\min}$-regular tree. Our proofs for the case of an underlying regular tree are based largely on the observation made in Corollary~\ref{cor:center_change}, which states that the Jordan center does not move as long as the two deepest subtrees continue to grow in depth. If this property holds true for the embedded 
$d_{\min}$-regular tree, then it is also true for the original irregular tree graph and thus, we have that the Jordan center will not move here as well. Similar to the previous section, we simulate the IC model with $p=0.4$ on an irregular tree where each node has a degree uniformly chosen between $3$ and $4$. These results are shown in Figures \ref{fig:Jordan_IC_dist_4} and \ref{fig:Jordan_IC_centre_4}.

Also, while we have considered infection spreading models for the growth of the random tree, one can consider other evolution models as well. For example, we can consider a tree graph growing according to the Preferential Attachment model \cite{BA} in which each new node connects to existing nodes with probability proportional to the degrees of the nodes in the previous time step. We track the Jordan center of the graph at each timestep as it grows. Some preliminary simulation results reported in Figures~\ref{fig:Jordan_PA_dist} and \ref{fig:Jordan_PA_centre}  do seem to indicate that the Jordan center is persistent in this model as well and indeed remains very close to the origin node.

\begin{figure}
        \begin{subfigure}[b]{0.49\textwidth}
                \includegraphics[width=\linewidth]{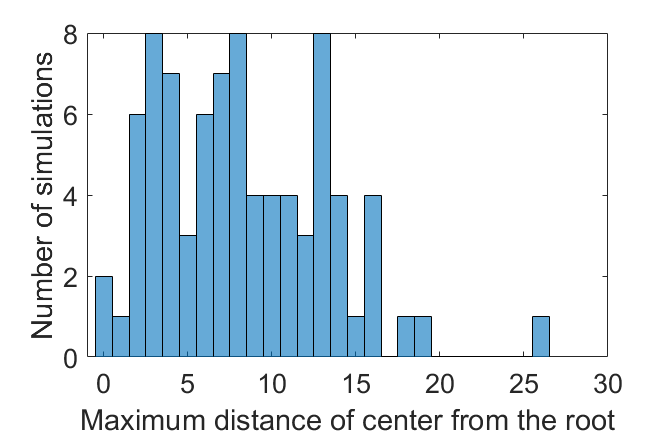}
                \caption{Distance from the root}
                \label{fig:Balancedness_IC_dist}
        \end{subfigure}%
        \begin{subfigure}[b]{0.49\textwidth}
                \includegraphics[width=\linewidth]{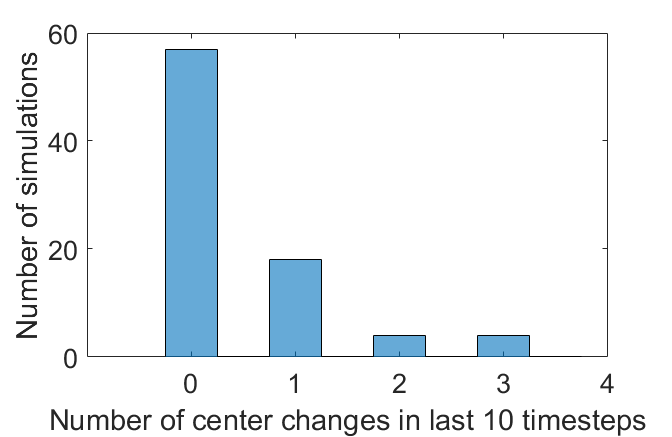}
                \caption{Number of center changes}
                \label{fig:Balancedness_IC_centre}
        \end{subfigure}%
        \caption{Balancedness Center in the IC Model}
\end{figure}

2) \textbf{Other graph centers: } While we have only considered the Jordan center in this paper, there are various other popular notions of a graph center in the literature, such as the distance center, betweenness center etc. An obvious research direction is to explore the question of persistence for different centers under various growth models. Recently, \cite{Jog1} considered scoring vertices in a tree graph based on the maximum size of the  subtrees rooted at the neighbors of the considered vertex, and then choosing the vertex with the minimum score as the tree center. We will refer to this center as the ``balancedness center". They proved the persistence of this center in both the uniform as well as preferential attachment tree growth models, as well as the (continuous-time) SI infection spread model on a regular tree. In each of these models, one new node is added to the tree at any time. On the other hand, for the IC and discrete SI models studied in this paper, several nodes can simultaneously be added to the tree and this can introduce much more variation in the movement of the balancedness center. For example, Figures~\ref{fig:Balancedness_IC_dist} and \ref{fig:Balancedness_IC_centre} indicate that the balancedness center in this case demonstrates much more movement and it will be interesting to check if the balancedness center is indeed persistent in this case. 




\begin{thebibliography}{99}

\bibitem{Borgatti}Borgatti, Stephen P. ``Centrality and network flow." \textit{Social networks} 27.1: 55-71, 2005.

\bibitem{Jackson} Jackson, Matthew O. \textit{Social and Economic Networks.} Princeton University Press, 2010.

\bibitem{Dawson} Dawson, Donald A. ``Introductory lectures on stochastic population systems." \textit{arXiv preprint arXiv:1705.03781} (2017).

\bibitem{Fanti1} Fanti, Giulia, Peter Kairouz, Sewoong Oh, Kannan Ramchandran, and Pramod Viswanath. ``Rumor source obfuscation on irregular trees". In \textit{ACM SIGMETRICS International Conference on Measurement and Modeling of Computer Science}, pp. 153-164, 2016.

\bibitem{Fanti2} Fanti, Giulia, Peter Kairouz, Sewoong Oh, Kannan Ramchandran, and Pramod Viswanath. ``Metadata-conscious anonymous messaging." In \textit{International Conference on Machine Learning}, pp. 108-116, 2016.

\bibitem{GW} Haccou, Patsy, Peter Jagers, and Vladimir A. Vatutin. \textit{Branching processes: variation, growth, and extinction of populations}. Cambridge university press, 2005.

\bibitem{Jog1} Jog, Varun, and Po-Ling Loh. ``Persistence of centrality in random growing trees." \textit{Random Structures \& Algorithms}, 2018.

\bibitem{Kingman} Kingman, John F. C. ``The first birth problem for an age-dependent branching process." \textit{The Annals of Probability} 790-801, 1975.

\bibitem{mcdiarmid} McDiarmid, Colin. ``Minimal positions in a branching random walk." \textit{The Annals of Applied Probability} 5.1: 128-139, 1995.

\bibitem{SIR} Zhu, Kai, and Lei Ying. ``Information source detection in the SIR model: A sample-path-based approach." \textit{IEEE/ACM Transactions on Networking } 24.1: 408-421, 2016.

\bibitem{BA} Barabasi, Albert-Laszlo, and Reka Albert. ``Emergence of scaling in random networks." \textit{Science} 286.5439 : 509-512, 1999.

\bibitem{SZ} Shah, Devavrat and Tauhid Zaman. ``Rumors in a network: Who's the culprit?" \textit{ IEEE Transactions on Information Theory,} 57(8):5163-5181, 2011.

\bibitem{BDL} Bubeck, Sebastien, Luc Devroye, and Gabor Lugosi. ``Finding Adam in random growing trees. Random Structures and Algorithms". \textit{Random Structures \& Algorithms}, 50.2: 158-172, 2017.

\bibitem{Slater} Slater, Peter J. ``Maximin facility location". \textit{Journal of National Bureau of Standards B}, 79:107-115, 1975.

\bibitem{Tan} Tan, Chee Wei, Pei-Duo Yu, Chun-Kiu Lai, Wenyi Zhang, and Hung-Lin Fu. ``Optimal detection of influential spreaders in online social networks." \textit{Annual Conference on Information Science and Systems (CISS)}, pp. 145-150, 2016.

\bibitem{Hwang} Hwang, Woochang, Young-rae Cho, Aidong Zhang, and Murali Ramanathan. ``Bridging centrality: identifying bridging nodes in scale-free networks." In \textit{ACM SIGKDD International conference on Knowledge discovery and data mining}, pp. 20-23, 2006.


\bibitem{Galashin} Galashin, Pavel. ``Existence of a persistent hub in the convex preferential attachment model". \textit{Probability and Mathematical Statistics}, pp. 59-74, 2016.

\bibitem{Mitchell}  Mitchell, Sandra L. ``Another characterization of the centroid of a tree." \textit{Discrete Mathematics}, 24(3):277-280, 1978.

\bibitem{Jog2} Jog, Varun, and Po-Ling Loh. ``Analysis of centrality in sublinear preferential attachment trees via the Crump-Mode-Jagers branching process." \textit{IEEE Transactions on Network Science and Engineering} 4, no. 1: 1-12, 2017.

\bibitem{KhimLoh} Khim, Justin, and Po-Ling Loh. ``Confidence sets for the source of a diffusion in regular trees." \textit{IEEE Transactions on Network Science and Engineering} 4, no. 1: 27-40, 2017.

\bibitem{Ying2} Zhu, Kai, and Lei Ying. ``Information source detection in networks: Possibility and impossibility results." In \textit{Annual IEEE International Conference on Computer Communications (INFOCOM)} pp. 1-9. IEEE, 2016.

\bibitem{JC} Hedetniemi, S. Mitchell, E. J. Cockayne, and S. T. Hedetniemi. ``Linear algorithms for finding the Jordan center and path center of a tree." \textit{Transportation Science} 15, no. 2: 98-114, 1981.

\bibitem{Handler} Handler, Gabriel Y. ``Minimax location of a facility in an undirected tree graph." \textit{Transportation Science} 7, no. 3: 287-293, 1973.
\end{thebibliography}

\section{Appendix}
\label{sec:Appendix}

We use Theorems $1$ and $2$ from \cite{mcdiarmid} in order to establish the concentration inequalities used in Section \ref{sec:Jordan_SI}. Here, we show that our problem setting satisfy all the necessary conditions for these theorems to be applied.
 
\noindent Let $Z_n(t)$ denotes the number of nodes $v$ in the $n^{th}$ generation (nodes which are at a distance of $n$ from the root node), which are born before time $t$. $Z_n = \sup_{t \rightarrow \infty} Z_n(t)$. Since our underlying tree model is a $d+1$-regular tree, we have $Z_n = d^n$. $B_n$ is defined as the first time of a birth in the $n^{th}$ generation. The time of infection of the nodes in the $n^{th}$ generation are denoted by $z_{n1}, z_{n2}, \cdots$. $F(t) = \mathbb{E}[Z_1(t)]$ and $\alpha = \inf \{t : F(t) > 0\}$. \\

\noindent Since we are working with the SI model, Definition \ref{def:SI_Model}, where the infection spread from one node to its neighbour is exponentially distributed with mean $1$, we can see that $\alpha = 0$. Now, define:
\begin{align}
\phi (\theta) = \mathbb{E}[ \sum_r e^{-\theta z_{1r}} ]
\end{align}
Since $z_{1r} \sim exp(1)$ we have:
\begin{align}
\phi (\theta) = \frac{d}{1+\theta}
\end{align}
Again, define:
\begin{align}
\mu (a) = \inf \{ e^{\theta a} \phi (\theta) : \theta \geq 0 \}
\end{align}
On substituting $\phi(\cdot)$ we get:
\begin{align}
\mu(a) = ad e^{1-a}
\end{align}
Time constant $\gamma$ is defined as:
\begin{align}
\gamma = \inf \{ a : \mu(a) \geq 1\}
\end{align}
Since $\mu(0) = 0$ and $\mu(1) = d$, we have by the continuity of $\mu(\cdot)$ that 
\begin{align}
0 < \gamma < 1
\end{align}

\noindent We have that the conditions of Theorem 2 \cite{mcdiarmid} are satisfied which gives for some constants $\gamma > 0, c_1, \delta>0$
\begin{align}
\mathbb{P}( B_n \leq \gamma n + c_1 \log n - x) \leq e^{-\delta x}
\end{align}
Since $Z_1 = d < \infty$, we apply part (b) of the theorem on $B_{n - c_3 \log n}$ to get:
\begin{align}
\mathbb{P}(B_{n - c_3 \log n} \geq \gamma ({n - c_3 \log n}) + c_2 \log ({n - c_3 \log n}) + x) \leq e^{-\delta x}
\end{align}

\end{document}